\documentclass[12pt]{article}
\usepackage[a4paper,margin=1.15in]{geometry}
\usepackage{graphicx,amsmath,amsthm,cite,xcolor}
\newtheorem{theorem}{Theorem}
\newtheorem{definition}[theorem]{Definition}

\newtheorem{lemma}[theorem]{Lemma}

\begin{document}
\title{A Game of Cops and Robbers on Graphs with Periodic Edge-Connectivity}
\author{Thomas Erlebach \and
Jakob T. Spooner}
\date{\small
School of Informatics, University of Leicester, England\\ \vspace{5pt}
 \{te17, jts21\}@leicester.ac.uk}
\maketitle
\begin{abstract}
This paper considers a game in which a single cop and a single robber take turns moving along the edges of a given graph $G$. If there exists a strategy for the cop which enables it to be positioned at the same vertex as the robber eventually, then $G$ is called cop-win, and robber-win otherwise. We study this classical combinatorial game in a novel context, broadening the class of potential game arenas to include the edge-periodic graphs. These are graphs with an infinite lifetime comprised of discrete time steps such that each edge $e$ is assigned a bit pattern of length~$l_e$, with a 1 in the $i$-th position of the pattern indicating the presence of edge $e$ in the $i$-th step of each consecutive block of $l_e$ steps. Utilising the already-developed framework of reachability games, we extend existing techniques to obtain, amongst other results, an $O(\textsf{LCM}(L)\cdot n^3)$ upper bound on the time required to decide if a given $n$-vertex edge-periodic graph $G^\tau$ is cop or robber win as well as compute a strategy for the winning player (here, $L$ is the set of all edge pattern lengths $l_e$, and $\textsf{LCM}(L)$ denotes the least common multiple of the set $L$).
Separately, turning our attention to edge-periodic cycle graphs, we give proof of a $2\cdot l \cdot \textsf{LCM}(L)$ upper bound on the length required by any edge-periodic cycle to ensure that it is robber win, where $l = 1$ if $\textsf{LCM}(L) \geq 2\cdot \max  L  $, and $l=2$ otherwise. Furthermore, we provide lower bound constructions in the form of cop-win edge-periodic cycles: one with length $1.5 \cdot \textsf{LCM}(L)$ in the $l=1$ case and one with length $3\cdot \textsf{LCM}(L)$ in the $l=2$ case.
\end{abstract}

\section{Introduction}
Pursuit-evasion games are games played between two teams of players, who take turns moving within the confines of some abstract arena. Typically, one team -- the \textit{pursuers} -- are tasked with catching the members of the other team -- the \textit{evaders} -- whose task it is to evade capture indefinitely. The study of such games has led to their application in a number of real-world scenarios, one widely-studied example of which would be their application to the problem of guiding robots through real-world environments \cite{Chung11}. From a theoretical standpoint, other variants of the game have been studied for their intrinsic links to important graph parameters; for example, in one particular variant in which each pursuer can, in a single turn, move to an arbitrary vertex of the given graph~$G$, it is well known that establishing the number of pursuers it takes to catch one evader also establishes the treewidth of $G$ \cite{Seymour93}.

The variant most closely resembled by the one considered in this paper was first studied separately by Quilliot \cite{Quilliot78}, and by Nowakowski and Winkler \cite{Nowakowski83}, as the discrete \textit{Cops and Robbers} game. In essence, the games these authors considered were the same: one cop (pursuer) and one robber (evader) take turns moving across the edges (or remaining at their current vertex) of a given graph $G$, with the cop aiming to catch the robber, and the robber attempting to avoid capture. (By `catching the robber' we mean that the cop is able to occupy the same vertex as the robber within $G$.) In this paper, we consider a variant of the \textit{Cops and Robbers} game with an almost identical set of rules to the one considered in \cite{Quilliot78, Nowakowski83}, but broaden the class of viable game arenas to include the \textit{edge-periodic} graphs \cite{Casteigts11}. As such, we call the game \textit{edge-periodic Cops and Robbers}, or \textsf{EPCR} for short. Informally, such graphs can be thought of as traditional static graphs equipped with an additional function, mapping each edge $e$ to a \textit{pattern} of length $l_e$ that dictates in which time steps $e$ is present within each consecutive period of $l_e$ steps. (Formal definitions of \textsf{EPCR} as well as the class of edge-periodic graphs are given in Section 2.) The class of edge-periodic graphs could also be considered a subclass of so-called \textit{temporal graphs}~\cite{MS18}.

The blanket term `temporal graphs' covers a broad, and relatively new area of interest to mathematicians and computer scientists, which looks to examine the inherent properties of graphs that have had incorporated into their combinatorial structure some element of time-variance. As far as we are aware,  cops and robbers type games (in fact, pursuit-evasion games in general) have not yet been studied in the context of temporal graphs.  

The contribution of this paper is twofold: in Section 3, we consider the problem of deciding, given an edge-periodic graph $G^\tau$, whether a game of edge-periodic cops and robbers played on $G^\tau$ is won by the cop, or won by the robber. We exploit a connection (that was previously noted in \cite{Kehagias142}) between the game of cops and robbers (the one considered in \cite{Quilliot78, Nowakowski83}) and \textit{reachability games}, which are a well studied class of 2-player games with strong connections to formal verification. Utilising reachability games as a framework, we extend existing techniques in order to solve the one cop, one robber variant of cops and robbers on a wider collection of graphs: specifically, what we show is that, given an edge-periodic graph $G^\tau$ whose edges $e$ are each labelled with patterns of length $l_e$, it is possible to decide in $O(\textsf{LCM}(L) \cdot n^3)$ time whether cop or robber wins on $G^\tau$ (where $\textsf{LCM}(L)$ denotes the least common multiple of the set of all edge-pattern lengths in $G^\tau$). We also show that by applying the general form of this result to some subclasses of the class of edge-periodic graphs, we are able to obtain polynomial upper bounds on the amount of time taken to determine the winner. Further, we show that it is possible to construct a strategy for the winning player using an algorithm with the same $O(\textsf{LCM}(L) \cdot n^3)$ time bound. 

In Section 4, we consider a subclass of the edge-periodic graphs, the members of which have underlying cycle graphs. We provide proof of an upper bound of $2 \cdot l \cdot \textsf{LCM}(L)$ on the length required by any such cycle $C^\tau$ in order to guarantee that it is robber-win, where $l = 1$ if $\textsf{LCM}(L) \geq 2\cdot \max  L  $, and $l = 2$ otherwise. Also provided within this section are lower bound constructions showing that there exist cop-win edge-periodic cycles of length $\frac{3}{2} \cdot \textsf{LCM(L)}$ and $3\cdot \textsf{LCM}(L)$ in the $l=1$ and $l=2$ cases, respectively. 
 
\subsection{Related Work}
The introduction of pursuit-evasion type combinatorial games is most often attributed to Torrence D. Parsons, who studied a problem in which a team of rescuers search for a lost spelunker in a circular cave system \cite{Parsons78}. By representing the cave as a cycle graph, he showed that one rescuer is not enough to guarantee that the spelunker is found, but that two are. In a similar vein to the work of Parsons, the \textit{Cop and Robber} problem, in which one cop attempts to catch a robber in a given graph $G$, was introduced independently by both Quilliot \cite{Quilliot78}, and by Nowakowski and Winkler \cite{Nowakowski83}. Their papers characterise precisely those graphs for which one cop is enough to guarantee that the robber is caught. Aigner and Fromme \cite{Aigner84} considered a generalised variant of the game, in which $k$ cops attempt to catch a single robber; their paper introduced the notion of the \textit{cop-number} of a graph, i.e., the minimum number of cops required to guarantee that the robber is caught.

In \cite{Hahn06, Bonato15, Berarducci93}, the authors develop reductions from the standard game of cops and robbers to a \textit{directed game graph}, and specify algorithms that can decide, for a given graph, whether cop or robber wins. In \cite{Kehagias142}, Kehagias and Konstantinidis note a connection between the formulations of \cite{Hahn06, Bonato15, Berarducci93} and reachability games. Reachability games are a well-studied class of 2-player \textit{token-pushing} games, in which two players push a token along the edges of a directed graph in turn -- one with the aim to push the token to some vertex belonging to a prespecified subset of the graph's vertex set, and the other with the aim to ensure the token never reaches such a vertex \cite{Gradel02}. It is well known that the winner of a reachability game played on a given directed graph $G$ can be established in polynomial time \cite{Berwanger12, Gradel02}. For more information regarding cops and robbers/pursuit-evasion games, as well as their connection to reachability games, we refer the reader to \cite{Gradel02, Berwanger12, Kehagias142, Kehagias14, Bonato11, Patsko17, Chung11}.

In this paper, we consider the game of cops and robbers within the context of a \textit{temporal graph} model. Temporal graphs are a relatively new object of interest, and incorporate an aspect of time-variance into the combinatorial structure of traditional static graphs \cite{MS18}. One previously considered way of viewing a temporal graph $\mathcal{G}$ is as a sequence of $L$ subgraphs of a given \textit{underlying graph} $G$ (where $L$ is the \textit{lifetime} of the graph) \cite{Michail16}, with each subgraph indexed by the time steps $t \in [L]$. For problems within this model, it is often natural to assume that each subgraph $G_t$ in all time steps $t \in [L]$ is connected \cite{Michail16}. The edge-periodic graphs considered in this paper differ in that this connectivity assumption is dropped -- similar graphs were introduced in \cite{Casteigts11}. For a detailed account of the theory of temporal graphs thus far, as well as the various models/problems that have been studied, we refer the reader to, for example, \cite{Casteigts18, Michail16, MS18}.

\section{Graph Model and Game Rules}
For any positive integer~$k$ we write $[k]$ for the set $\{0,1,\ldots,k-1\}$.
We begin with our definition of \emph{edge-periodic} graphs:

\begin{definition}[Edge-periodic graph $G^{\tau}$]\label{def1}
	We define an edge-periodic graph $G^\tau = (V, E, \tau)$ to be a temporal graph with
	underlying graph $G=(V,E)$ and
	infinite lifetime, and an additional relation $\tau: E \rightarrow \{0,1\}^*$, which maps each edge $e \in E$ to a \textit{pattern} $\tau(e) = b_e(0),b_e(1),...,b_e(l_e - 1)$ of length $l_e$. Each $\tau(e)$ consists of $l_e$ boolean values, such that $e$ is present in a time step $t$ if and only if $b_e(t \bmod l_e) = 1$; otherwise, $e$ is not present in time step $t$. Additionally, assume that for any edge $e \in E(G^\tau)$, $b_e(i) = 1$ for at least one $i \in [l_e]$; this implies that every edge is present at least once in any period of $l_e$ time steps.
\end{definition}
(We note that any $G^\tau$ whose edges each have patterns consisting of all $1$'s can be treated, under our model, as equivalent to the traditional static graph.) In this paper, we consider a game of cop(s) and robber(s) identical in its rule set to the one introduced in \cite{Quilliot78} and \cite{Nowakowski83} (in particular, the variant with 1 cop and 1 robber), but broaden the class of viable game-arenas to include all possible edge-periodic graphs $G^\tau$.  We call the resulting game \textit{edge-periodic cop(s) and robber(s)}, or \textsf{EPCR} for short -- the following paragraph specifies its rule set, as well as the separate winning conditions for \textsf{C} and \textsf{R}:

\textbf{Rules of \textsf{EPCR}.} Initially, the two players (cop \textsf{C} and robber \textsf{R}) each select a start vertex on a given edge-periodic graph $G^\tau$. \textsf{C} chooses first, followed by \textsf{R}, whose choice is made in full knowledge of \textsf{C}'s choice. After the start vertices have been chosen, in each subsequent time step $t \geq 0$, players take alternating turns moving over an edge in the graph that is incident to their current vertex (or indeed choosing to remain at their current vertex), following the convention that in any particular time step, \textsf{C} moves first, in full knowledge of \textsf{R}'s position, followed by \textsf{R}; again, \textsf{R}'s move is made with full knowledge of the move that \textsf{C} just made. In all subsequent turns, each player makes their moves in full knowledge of the other player's previous move. In line with Definition~\ref{def1}, whenever \textsf{C} or \textsf{R} are situated at a vertex $v \in V(G^\tau)$ during some time step $t$ and it is their turn to make a move, they may only traverse those edges $\{v, u\}$ such that $b_{\{v, u\}}(t \bmod l_{\{v,u\}}) = 1$. The game terminates only when, at the end of either player's move, \textsf{C} and \textsf{R} are situated at the same vertex in $G^\tau$. If there exists a strategy for \textsf{C} which ensures that the game terminates, we say that $G^\tau$ is \textit{cop-win}. Otherwise, there must exist a strategy for \textsf{R} which enables infinite evasion of \textsf{C}; in this case we call $G^\tau$ \textit{robber-win}.

\section{Determining the Winner of a Game of \textsf{EPCR}}
\label{sec:main}
Within this section, our aim is to prove the following theorem (throughout the following, we use $\textsf{LCM}(A)$ to denote the lowest common multiple of a set of integers~$A$):

\begin{theorem}\label{thm2} Let $G^\tau$ be an edge-periodic graph of order $n$, and let $L = \{l_e : e \in E(G^\tau)\}$. Then, it can be decided in $O(\textsf{LCM}(L) \cdot n^3)$ time whether or not $G^{\tau}$ is cop-win.
\end{theorem}

The proof relies primarily on a transformation from a given edge-periodic graph $G^{\tau}$ to a finite \emph{directed game graph} $G'$. The transformation is such that the playing of an instance of \textsf{EPCR} on $G^\tau$ is, in some sense, equivalent to the playing of a 2-player token-pushing game (specifically a \textit{reachability game}, which will be properly defined in due course) played on $G'$. To establish this equivalence, we need a way of translating a particular state of an instance of \textsf{EPCR} played on $G^\tau$ to a corresponding state in the reachability game played on $G'$. To this end, the following definition properly introduces the notion of a \textit{position} in a game of \textsf{EPCR}, played on an edge-periodic graph $G^{\tau}$.

\begin{definition}[Position in $G^{\tau}$]\label{def3}
The current state of a game of \textsf{EPCR} played on an edge-periodic graph $G^{\tau}$ is determined by 4 individual pieces of information: (1) the vertex currently occupied by \textsf{C}; (2) the vertex currently occupied by \textsf{R}; (3) the player whose turn it is to move; and (4) the current time step t. To this end, we can define a position in $G^{\tau}$ to be a $4$-tuple, \[P = (c_P, r_P, s_P, t_P),\] where $c_P \in V(G^{\tau})$ is \textsf{C}'s current vertex, $r_P \in V(G^{\tau})$ is \textsf{R}'s current vertex, $s_P \in \{\textsf{C} , \textsf{R}\}$ is the player whose turn it is to move next, and $t_P$ is the current time step. 
\end{definition}
We call any position $P$ such that $c_P = r_P$ a \emph{terminating position}, since this indicates that both players are situated on the same vertex and hence \textsf{C} has won. We now formally introduce the notion of \textit{reachability games} \cite{Gradel02}:

\begin{definition}[Reachability game $G'$]\label{def4} A reachability game is a directed graph $G'$, given as a 3-tuple: \[G' = (V_0 \cup V_1, E', F),\] where $V_0 \cup V_1$ is a partition of the state set $V'$; $E' \subseteq V' \times V'$ is a set of directed edges; and $F \subseteq V'$ is a set of final states. 
\end{definition}
The game is played by two opposing players, \textsf{Player 0} and \textsf{Player 1}; $V_0$ and $V_1$ are the (disjoint) sets of \textsf{Player 0}/\textsf{Player 1} owned nodes, respectively. One can imagine a token being placed at some initial vertex (call it $v_0$) at the start of the game. Depending on whether $v_0 \in V_0$ or $v_0 \in V_1$, we can then imagine the corresponding player selecting one of the outgoing edges of $v_0$, and pushing the token along that edge. When the token arrives at the next vertex, the corresponding player then selects an outgoing edge and pushes the token along it. This process then continues --  such a sequence of moves constitutes a \textit{play} of the reachability game on $G'$. Formally, a play $\phi = v_0,v_1,...$ is a (possibly infinite) sequence of vertices in~$V'$, such that $(v_i, v_{i+1}) \in E'$ for all $i \geq 0$. We say that a play $\phi$ is \textit{won} by \textsf{Player 0} if there exists some $i$ such that $v_i \in F$. Otherwise, $\phi$ is of infinite length and for no $i$ is $v_i \in F$; in this latter case, we say that $\phi$ is won by \textsf{Player 1}. 

\subsection{Transformation}
We now detail our transformation from a given edge-periodic graph $G^\tau$ to a reachability game $\beta(G^\tau)$: let $\beta$ be a transformation function that takes as argument a given edge-periodic graph $G^\tau$, so that the notation $\beta(G^\tau)$ denotes the game graph $G'$ on which each play of a reachability game corresponds to a sequence of moves performed by \textsf{C} and \textsf{R} in a game of \textsf{EPCR} on $G^\tau$. As before, we denote by $\textsf{LCM}(A)$ the least common multiple of the set of integers $A$. Further, let $L = \{l_e : e \in E(G^\tau)\}$. We let $\beta(G^\tau) := G' = (V', E', F)$, and go on to define its individual components below:

\textbf{State set $V'$.} We define the state set (i.e., vertex set) of our directed game graph $\beta(G^\tau)$ to be a set of 4-tuples, each corresponding to a position in the game of \textsf{EPCR} on $G^\tau$ as follows: \[V' = \{(c, r, s, t) : c,r \in V(G^\tau), s \in \{\textsf{C}, \textsf{R}\}, \text{ and } t \in [\textsf{LCM}(L)]\}.\] Keeping in line with Definition~\ref{def4}, we also let $V_0 := \{(c, r, s, t) \in V' : s = \textsf{C}\} \text{ and } V_1 := \{(c, r, s, t) \in V' : s = \textsf{R}\}$ be the sets of \textsf{Player 0} (or \textsf{Cop}) owned nodes,  and \textsf{Player 1} (or \textsf{Robber}) owned nodes, respectively. In order to justify this formulation (in particular, to justify the range of variable $t$ -- the current time step component of any state $S \in V'$), recall from Definition~\ref{def3} that a position $P$ of a game of \textsf{EPCR} consists of a cop vertex $c_P$, a robber vertex $r_P$, a variable $s_P \in \{\textsf{C}, \textsf{R}\}$ which represents the player whose turn it is to move next, and a current time step $t_P$. We wish to capture with our finite directed game graph~$G'$ all possible positions~$P$ of \textsf{EPCR} on $G^\tau$. Since, by Definition~\ref{def1}, the lifetime of $G^\tau$ is infinite, we cannot simply create a state $S \in V(G')$ corresponding to each possible position $P$; the infinite number of time steps would result in an infinite game graph $G'$. It is not hard to see that in the $\textsf{LCM}(L)$-th step, all edge patterns will finish (since $l_e$ divides $\textsf{LCM}(L)$ for all $e \in E(G^\tau)$), and so in the following step, all patterns will restart. As such, we can view the temporal structure of our edge set as an infinitely repeating pattern, and by letting $t$ range over the integers in $[\textsf{LCM}(L)]$ we are able to properly capture this structure using only a finite number of states. 

\textbf{Edge set $E'$.} In order to construct the edge set $E' \subseteq (V_0 \times V_1) \cup (V_1 \times V_0)$, we consider all pairs of states $S = (c, r, s, t)$ and $S' = (c', r', s', t')$ such that $S \neq S'$ and $S, S' \in V'$. We can then see $E'$ as the set of edges such that $(S, S') \in E'$ if and only if the states $S$ and $S'$ satisfy all of the conditions below:

\begin{itemize}
	\item[(1)]{$\big(s = \textsf{C} \land s' = \textsf{R}\big) \lor \big(s=\textsf{R} \land s'=\textsf{C}\big)$,}
	\item[(2)]{$s = \textsf{C} \implies \big(c = c' \lor \{c, c'\} \in E(G^\tau)\big) \land (r = r') \land (t' = t)$,}
	\item[(3)]{$s = \textsf{R} \implies \big(r = r' \lor \{r, r'\} \in E(G^\tau)\big) \land \big(c = c'\big) \newline \land \big(t' \in [\textsf{LCM}(L)] \text{ satisfies } t' = (t+1) \bmod \textsf{LCM}(L)\big)$,}
	\item[(4)]{if $s = \textsf{C}$ and $c\neq c'$, then $b_{\{c, c'\}}(t\bmod l_{\{c, c'\}}) = 1$,}
	\item[(5)]{if $s = \textsf{R}$ and $r\neq r'$, then $b_{\{r, r'\}}(t\bmod l_{\{r, r'\}}) = 1$.}
\end{itemize}

Condition (1) ensures that any sequence of moves constituting a play in $G'$ alternate between \textsf{C} and \textsf{R}, which keeps in line with the rules of \textsf{EPCR}. Condition (2) ensures that any state $S'$, reachable in one move from a \textsf{C}-owned state $S$, is such that $c'$ is adjacent to $c$ in $G^\tau$ (or is in fact~$c$, indicating that \textsf{C} has waited at the current vertex); that the robber's new vertex $r$ does not change; and, that $t' = t$, keeping in line with the rule stating that \textsf{C} moves first in any given time step, followed by \textsf{R} who must also make a move in step $t$. On the other hand, Condition (3) ensures that, once $\textsf{R}$ pushes the token from some $\textsf{R}$-owned state $S$, \textsf{R}'s new vertex $r'$ is adjacent to $r$ in $G^\tau$ (or equal to $r$); that \textsf{C}'s vertex remains the same, and that the state $S'$ to which the token is pushed onto is a state in which the current time step is advanced by one. Conditions (4)-(5) ensure that both players can only make moves across edges that are incident to their current vertex, as well as present in the current time step; on the other hand, they also ensure that either player always has the ability to remain at their current vertex in any step $t$ if they should choose to do so.

\textbf{Set of final states $F$.} Let $F = \{ (c,r,s,t) \in V' : c = r \}$, so that the set of final states consists of all those states that correspond to a position in $G^\tau$ such that \textsf{C} is positioned on the same vertex as \textsf{R}. This adheres to the rules of \textsf{EPCR}, which state that the game terminates only when this condition is met by the current position.

\subsection{Proof of Theorem \ref{thm2}}
We first introduce the elements of the theory of reachability games that are required for the proof of Theorem~\ref{thm2}, starting with the definition of the \textit{attractor set}:
\begin{definition}[Attractor set $Attr(F)$ \cite{Berwanger12}]\label{def5}
	The sequence $({Attr}_i(F))_{i\geq 0}$ is recursively defined as follows:
\begin{eqnarray*}
	{Attr}_0(F) & = & F \\ 
	{Attr}_{i+1}(F) & = & {Attr}_i(F) \cup \{v \in V_0 \mid \exists (v,u) \in E' : u \in {Attr}_i(F)\} \cup  \\
				  & \phantom{   } & \{v \in V_1 \mid \forall(v, u) \in E' : u \in {Attr}_i(F) \}
\end{eqnarray*}
	 We can see that the sets $Attr_i(F)$, as defined above, are a sequence of subsets of $V'$ that are monotone with respect to set-inclusion. We then let \[{Attr}(F) = \bigcup_{i \geq 0} {Attr}_i(F).\] Since $G'$ is finite, we are able to view the set $Attr(F)$ as the least fixed point of the sequence $({Attr}_i(F))_{i\geq 0}$.
\end{definition}

From Definition \ref{def5} it follows by induction that, from those states $S \in {Attr}_i(F) \cap V_0$ such that $i \geq 1$ and $S \notin Attr_j(F)$ for any $j < i$, \textsf{Player 0} is able to force the sequence of play into some state $S_F \in F$ within $i$ moves, by selecting for each such $S$ a successor state $S'$ such that $(S, S') \in E'$ and $S' \in Attr_{i-c}$ for some $c \geq 1$. On the other hand we have that, from any state $S \in Attr_i(F) \cap V_1$ (again, let $i \geq 1$ and $S \notin Attr_j(F)$ for any $j < i$), \textsf{Player 1} can not avoid forcing the sequence of play into a state $S' \in {Attr}_{i-c}$ (for some $c \geq 1$); from the definition of ${Attr}_{i}$ ($i \geq 0$), it again follows by induction that play will be forced into some state $S_F \in F$ in at most $i$ steps. This brings us to the following well-known result from the reachability games literature, which will be useful in proving Theorem~\ref{thm2}:

\begin{theorem}[Berwanger \cite{Berwanger12}]\label{thm6}
	In a given reachability game $G' = (V', E', F)$, $\textsf{Player 0}$ has a winning strategy from any state $S \in Attr(F)$, and $\textsf{Player 1}$ has a winning strategy from any state $S \in (V_0 \cup V_1) - Attr(F)$. 
\end{theorem}

Recall now that the transformation $\beta$ produces, from a given edge-periodic graph $G^\tau$, a directed game graph $\beta(G^\tau) = (V', E', F)$ such that there is a correspondence between every possible position in the game of \textsf{EPCR} on $G^\tau$ with some state in $V'$, and vice versa. Using the notation $S_P$ to refer to the state in $V'$ that corresponds to the position $P$ in the game of \textsf{EPCR} on $G^\tau$, we can compute the set $Attr(F)$ for our game graph $\beta(G^\tau)$ and thus, on invocation of Theorem \ref{thm6}, state the following lemma:

\begin{lemma}\label{lem7}
\textsf{Cop} can force a win from a position $P$ if and only if the state $S_P \in V(\beta(G^{\tau}))$ satisfies $S_P \in Attr(F)$.
\end{lemma}

Note that one consequence of Lemma \ref{lem7} is the following: In a game of \textsf{EPCR} on $G^\tau$ starting from a position $P$ such that $S_P \notin Attr(F)$, the robber can force the sequence of moves to never reach any state $S \in F$, and, as such, the game can be won by \textsf{R}.  

\begin{lemma}\label{lem8}
 An edge-periodic graph $G^{\tau}$ is cop-win if and only if there exists a vertex $v \in V(G^{\tau})$ such that $(v, r, \textsf{C}, 0) \in {Attr}(F)$ for all $r \in V(G^{\tau})$.
\end{lemma}
\begin{proof}
($\Rightarrow$) Assume not, so that $G^\tau$ is cop-win but there exists no vertex $v \in V(G^\tau)$ such that $(v, r, \textsf{C}, 0) \in {Attr}(F)$ for all $r \in V(G^\tau)$. Then for every $v$, there exists at least one vertex $u$ such that the state $(v, u, \textsf{C}, 0) \notin Attr(F)$. Let \textsf{C} choose its start vertex $c$ as $v$ (i.e., sets $c = v$), and let \textsf{R} set $r = u$. Since \textsf{R} chooses $r$ in full knowledge of \textsf{C}'s choice of $c$, it follows that \textsf{R} can force the equivalent reachability game on $\beta(G^\tau)$ to begin from a state $S_{(c,u,\textsf{C},0)} \notin Attr(F)$, hence winning the reachability game regardless of \textsf{C}'s choice of $c$. Notice that this implies that there exists a winning strategy for $\textsf{R}$ in the game of $\textsf{EPCR}$ on $G^\tau$; this is a contradiction since, by assumption, $G^\tau$ is cop-win.

($\Leftarrow$) Assume \textsf{C} chooses $v$ as its start vertex, i.e., sets $c = v$. By doing so, the equivalent reachability game on $\beta(G^\tau)$ can be forced to start at some state $(v, r, \textsf{C}, 0) \in Attr(F)$ regardless of \textsf{R}'s choice of $r$, since $(v, r, \textsf{C}, 0) \in Attr(F)$ for all $r \in V(G^{\tau})$. Hence, regardless of \textsf{R}'s choice of $r$, \textsf{C} wins the reachability game on $\beta(G^\tau)$, and, as a result, can win the game of $\textsf{EPCR}$ on $G^\tau$ by picking its start vertex as $v$; the lemma follows.
\end{proof}

The proof of the main theorem will also make use of a further known result from the reachability games literature; for the following, let $G' = (V', E', F)$ be a given directed game graph.

\begin{theorem}[Gr\"adel et al.\cite{Gradel02}]\label{thm9}
There exists an algorithm which computes the set $Attr(F)$ in time $O(|V'| + |E'|)$. 
\end{theorem}

Given the above, all is in place for the proof of Theorem \ref{thm2}:

\begin{proof}[Proof of Theorem \ref{thm2}]
	Since $n = |V(G^\tau)|$, $\beta$~produces, given an edge-periodic graph $G^\tau$, a directed game graph $\beta(G^\tau) = (V', E', F)$, such that $|V'| = O(\textsf{LCM}(L)\cdot n^2)$. To see this, observe first that for an arbitrary position $P = (c_P, r_P, s_P, t_P)$ in a game of \textsf{EPCR} on $G^\tau$, there are $n$ ways to choose $c_P \in V(G^\tau)$, $n$ ways to choose $r_P \in V(G^\tau)$, and a further 2 ways to choose $s_P \in \{\textsf{C}, \textsf{R}\}$. By definition of the transformation function, $G'$ has states for time steps $t \in [\textsf{LCM}(L)]$ only, and so in total we have that $|V'| = 2\cdot \textsf{LCM}(L)\cdot n^2 = O(\textsf{LCM}(L)\cdot n^2)$, as claimed. Next, note that each state $S_P \in V'$ has at most $n$ edges leading away from it to other states. This is because in the corresponding position $P$ in the game of $\textsf{EPCR}$ on $G^\tau$, the player whose turn it currently is has at most $n$ choices of moves across edges -- at most $n-1$ edges leading to other vertices plus the choice of remaining at the current node. Since there are $O(\textsf{LCM}(L)\cdot n^2)$ states $S \in V'$, it follows that $|E'| = O(\textsf{LCM}(L)\cdot n^3)$.
	
Combining the above with the result of Theorem \ref{thm9}, we can conclude that the attractor set $Attr(F)$ (that is, the set of all states from which \textsf{Player 0}, i.e., \textsf{C}, has a winning strategy) of any graph $\beta(G^\tau)$ can be computed in time $O(\textsf{LCM}(L)\cdot n^3)$. By Lemma \ref{lem8}, we can then verify whether or not $G^\tau$ is cop-win by checking if there exists at least one vertex $v \in V(G^\tau)$ such that $(v, u, \textsf{C}, 0) \in Attr(F)$, for all $u \in V(G^\tau)$; if such a $v$ exists, the algorithm will return $\textsf{YES}$, otherwise the algorithm will return $\textsf{NO}$. Carrying out this check can clearly take at most $O(n^2)$ time, and the theorem follows.
\end{proof}

We also note, as a direct consequence of Theorem~\ref{thm2}, that as long as $\textsf{LCM}(L)$ is polynomial in $n$ and $\max{L}$, then the winner of a given graph $G^\tau$ can be decided in polynomial time. Furthermore, if the labels $l_e$ are bounded by some constant for all $e \in E(G^\tau)$, then the winner can be decided in $O(n^3)$ time.

As well as being able to decide whether or not a given edge-periodic graph is cop-win or not, we would like to be able to compute a strategy for the winning player of the game of \textsf{EPCR} on a given graph $G^\tau$. One common way to view a strategy for \textsf{Player i} ($i \in \{0, 1\}$), in a general infinite game played on a game graph $G = (V, E, F)$ (where $V := V_0 \cup V_1$), is as a partial function $\sigma : V^*V_i \rightarrow V$. Here, $V^*V_i$ can be seen as the set of all prefixes (of any play $\phi$ in $G$) that end in a state $S \in V_i$, with $\sigma$ dictating to \textsf{Player i} the appropriate move to play, based on the history of these prefixes. 

On the other hand, a \textit{memoryless} strategy can be viewed more simply -- as a partial function $\sigma : V_i \rightarrow V'$. Such a strategy $\sigma$ can be employed in games where a correct move for a player depends not on the entire state-history of some play (or a prefix of) $\phi$, but only on the current state. It is well-known that reachability games fall into this category \cite{Gradel02}; since \textsf{EPCR} reduces to a reachability game, we are thus able to make use of the following result from the literature:

\begin{theorem}[Berwanger \cite{Berwanger12}]\label{thm12}
	Given a reachability game $G' = (V', E', F)$, one can compute in $O(|V'| + |E'|)$ time a memoryless winning-strategy for $\textsf{Player 0}$ from any state $S \in Attr(F)$, and a memoryless winning-strategy for $\textsf{Player 1}$ from any state $S \in (V_0 \cup V_1) - Attr(F)$. 
\end{theorem}

As such, given a directed game graph $\beta(G^\tau) = (V', E', F)$ (with $V' := V_0 \cup V_1$), Theorem \ref{thm12} tells us that it suffices to compute, for the winning player, a memoryless winning strategy $\sigma_i : V_i \rightarrow V'$, with the value of $i \in \{0, 1\}$ depending on the winner of the reachability game on $\beta(G^\tau)$. The following theorem shows that it is possible to interpret any such $\sigma$ as a strategy for the winning player in the corresponding game of \textsf{EPCR} on $G^\tau$:

\begin{theorem}\label{thm13}
	Let $G^\tau$ be an arbitrary edge-periodic graph and $L = \{l_e : e \in E(G^\tau)\}$. Then, depending on whether $G^\tau$ is cop-win or not, one can compute in $O(\textsf{LCM}(L)\cdot n^3)$ time either a memoryless winning strategy enabling \textsf{C} to capture \textsf{R}, or a memoryless winning strategy enabling \text{R} to evade capture infinitely.
\end{theorem}

\begin{proof}
Let $b \in \{ \textsf{YES}, \textsf{NO} \}$ be the return value of the algorithm from Theorem \ref{thm2} when provided $G^\tau$ as input. First, we construct a strategy for the winning player of the equivalent reachability game $\beta(G^\tau) = (V', E', F)$, and go on to show how such a strategy can then be interpreted as a strategy for the corresponding game of \textsf{EPCR} on $G^\tau$.

First, consider the case in which $b = \textsf{YES}$. Then we know $G^\tau$ is cop-win and, by Lemma \ref{lem8}, we know that there exists some vertex $v$ such that $(v, u, \textsf{C}, 0) \in Attr(F)$ for all $u \in V(G^\tau)$. As such, the initial stage of our strategy for \textsf{C} should consist of computing such a vertex $v$, and setting the cop start vertex in $G^\tau$ to $c = v$. Now, using Theorem \ref{thm12}, we can compute a memoryless winning-strategy $\sigma_\textsf{C}$. The initial stage of identifying some vertex $v$ such that $(v, u, \textsf{C}, 0) \in Attr(F)$ for all $u \in V(G^\tau)$ takes $O(n^2)$ time, and the algorithm of Theorem \ref{thm12} takes time at most $O(|V'| + |E'|) = O(\textsf{LCM}(L)\cdot n^3)$; it follows that the overall construction of a \textsf{C}-strategy for the reachability game $\beta(G^\tau)$ takes $O(\textsf{LCM}(L)\cdot n^3)$ time. Such a strategy for $\beta(G^\tau)$ can then be interpreted as strategy for \textsf{C} in the game of \textsf{EPCR} on $G^\tau$ by first selecting start vertex $v$. From then onward, whenever it is $\textsf{C}'s$ turn, in order to establish the appropriate move to play given a current position $P = (c_P, r_P, \textsf{C}, t_P)$, \textsf{C} constructs from it a state $S_P$, and checks the $c'$ component of the state $\sigma_\textsf{C}(S_P) = (c', r', \textsf{C}, t_P)$. Finally, it is guaranteed that $c'$ is adjacent to $c_P$ (or possibly $c = c_P$) during $t_P$, due to the way the transformation from $G^\tau$ to $\beta(G^\tau)$ has been defined.

In the situation in which $b = \textsf{NO}$, then we know that $G^\tau$ is robber-win, and thus know by Lemma \ref{lem8} that for every $v \in V(G^\tau)$, there exists at least one vertex $u$ such that the state $(v, u, \textsf{C}, 0) \notin Attr(F)$. Thus, the initial stage of our strategy for \textsf{R} involves the construction of a mapping $\sigma^0_\textsf{R} : V(G^\tau) \rightarrow V(G^\tau)$ from each possible $v \in V(G^\tau)$ that \textsf{C} might choose as its start vertex, to a vertex $u$ satisfying the aforementioned non-membership condition. Application of Theorem \ref{thm12} then allows us to construct a memoryless winning strategy, $\sigma_\textsf{R}$, from all states $S \in (V' - Attr(F)) \cap V_0$. The construction of $\sigma^0_\textsf{R}$ involves checking, for each of $n$ possible start vertices $v$ for \textsf{C}, at most $n$ vertices $u$ in order to identify which combination of $v$ and $u$ satisfies $(v, u, \textsf{C}, 0) \notin Attr(F)$. Hence, this initial phase can take at most $O(n^2)$ time. Similar to before, the algorithm of Theorem \ref{thm12} can take at most $O(|V'| + |E'|) = O(\textsf{LCM}(L)\cdot n^3)$ time, and hence we have that the overall construction of a strategy for \textsf{R} can take at most $O(\textsf{LCM}(L)\cdot n^3)$ time, as claimed. Finally, interpreting $\sigma_\textsf{R}$ as a strategy for $\textsf{R}$ in the game of \textsf{EPCR} on $G^\tau$ is the same as for \textsf{C}; the only difference is the way in which the start vertex is selected -- \textsf{R} waits until $\textsf{C}$ has selected a start vertex $c$, and then chooses its own start vertex as $\sigma^0_\textsf{R}(c)$.
\end{proof}

We remark that Theorems~\ref{thm2} and~\ref{thm13} can be generalised to a setting with $k$~cops at the expense of increasing the algorithm's running time to $O(\textsf{LCM}(L)\cdot k\cdot n^{k+2})$. The idea is to fix an arbitrary ordering of the cops, and create $k+1$ layers of states during every time step $t \in [\textsf{LCM}(L)]$ (one for each of the $k$ cops' moves, followed finally by the robber's move). By allowing in each time step for the players to play their moves in this serialised fashion the resulting game graph would require $O(\textsf{LCM}(L) \cdot k)$ layers with $n^{k+1}$ states in each, with at most $n$ edges leading from every state to states in the following layer.

\section{An Upper Bound on the Length Required to Ensure an Edge-Periodic Cycle is Robber-Win}
In this section, we consider a restricted subclass of the edge-periodic graphs -- in particular, we consider the subclass of edge-periodic cycles $C^\tau$. We provide an upper bound on the length required of any edge-periodic cycle~$C^\tau$, to ensure that it is robber-win. First, we show that any edge-periodic infinite path whose edge-pattern lengths originate from a set of integers $L$ with finite size is robber-win, and second show how the strategy for such infinite paths can be adapted to the cycle case. Throughout the following, we refer to cop as \textsf{C} and robber as \textsf{R}.  We always assume that we are given an edge periodic cycle $C^\tau = (V, E, \tau)$, and that the set of integers $L$ contains the lengths of every bit pattern that $\tau$ maps at least one edge $e \in E$ to. Additionally, we write \textsf{LCM} to denote the least common multiple of the numbers in the set~$L$.

We first present a lemma for infinite paths, which will also allow us to handle the case in which the cop chases the robber around the cycle in a fixed direction.

\begin{lemma}\label{lem15} Let $P$ be an infinite edge-periodic path, $L = \{l_e : e \in E(P)\}$ and assume that $|L|$ is finite. Then, starting from any time step $t$, if $\textsf{LCM} = \max  L  $ then there exists a winning strategy for \textsf{R} from any vertex with distance at least $2\cdot \textsf{LCM}$ from \textsf{C}'s start vertex, and at least $\textsf{LCM}$ otherwise.\end{lemma}

\begin{proof}
First, notice that since we assume that $|L|$ is finite, so must be $\textsf{LCM}$. Let \textsf{C} pick its initial vertex $c_t \in P$. Let $\textsf{R}$'s initial vertex be denoted by~$r_t$, and assume without loss of generality that $r_t$ will be some vertex in $P$ that lies to the right of~$c_t$. As such, we will from here onward denote by $P$ the path starting at $c_t$ and extending infinitely to the right. 

Consider the set $L$ and its constituent elements. There are two cases -- either (1) there exists $x \in L$ such that $\max  L  $ is not a multiple of $x$ -- then, $\textsf{LCM} \geq 2\cdot \max  L  $, since it cannot be the case that $\textsf{LCM}= j\cdot \max  L  $ for any $j < 2$; or (2) for every $x \in L$, $\max  L   = x\cdot i$ for some integer $i \geq 1$; then, $\textsf{LCM} = \max  L  $. With this in mind, define $B=\textsf{LCM}$ if (1) holds and $B=2\cdot\textsf{LCM}$ if (2) holds. Now, let us define the \textit{strips} $S_i$ ($i \geq 1$) to be finite subpaths of $P$, such that for all edges $e \in S_i$, $e$ can first be traversed by \textsf{C} in some time step $t_e \in [t+(i-1)B, t+iB-1]$. 
Note that $B\geq 2\cdot\max  L  $ and hence each $S_i$ must contain at least two edges.
By convention, we call the leftmost and rightmost edges (vertices) of any $S_i$ its \textit{first} and \textit{last} edges (vertices), respectively.
Note also that the last vertex of $S_i$ and the first vertex of $S_{i+1}$ are one and the same, for all $i \geq 1$.

Assume from now on that \textsf{C} moves right whenever possible. It is safe to do so since, otherwise,
\textsf{C} may only be positioned at the same vertex or further left than when following this strategy. The strategy for \textsf{R} is as follows: pick $r_t$ to be the first vertex of $S_2$ and move right (i.e., away from \textsf{C}) whenever possible. 

We now demonstrate that \textsf{R}'s strategy is a winning one. Let $T^F_X(i)$ and $T^L_X(i)$ denote the first time step that player $X \in \{\textsf{C}, \textsf{R}\}$  is able to traverse the first/last edge of $S_i$, respectively. Note that $T^F_\textsf{C}(i) \geq t + (i-1)B$ and that $T^L_\textsf{R}(i) \leq t + (i-1)B -1$. Combining the two gives that $T^L_\textsf{R}(i) < T^F_\textsf{C}(i)$, which implies that $\textsf{C}$ can never catch $\textsf{R}$ in any step in which the edge leading to both player's right belongs to $S_i$.

We next show that \textsf{R} cannot be caught when the edge leading to \textsf{C}'s right belongs to $S_i$ and the edge leading to \textsf{R}'s right belongs to $S_{i+1}$. Let $M = \max  L  $ and 
recall that $T_\textsf{C}^F(i) \geq t + (i-1)B$. 
Since the strips $S_i$ are defined to consist of all edges crossed in the period $[t + (i-1)B, t + iB-1]$, and since $B \geq 2M$, it follows that at time $T_\textsf{R}^F(i+1) \leq t + (i-1)B + M - 1$, there is at least one more edge of $S_i$ that remains to be crossed by \textsf{C}. This gives that $T^L_\textsf{C}(i) > T^F_\textsf{R}(i+1)$ and yields the claim. Combining this with the earlier observation that $T^L_\textsf{R}(i) < T^F_\textsf{C}(i)$, it follows that there exists a strategy for \textsf{R} starting from the first vertex of $S_2$.

Finally, recall that when $\textsf{LCM} = \max  L  $, we have that each $S_i$ consists of at most $2\cdot \textsf{LCM}$ edges, otherwise it consists of at most $\textsf{LCM}$. Hence, there exists a winning strategy for \textsf{R} starting from some vertex $r_t \in P$ with distance at most $2\cdot \textsf{LCM}$ or \textsf{LCM} from $c_t$, depending on the condition satisfied by \textsf{LCM}. The lemma follows by noticing that the above strategy also works when $\textsf{R}$ is initially positioned at any vertex further to the right than the first vertex of $S_2$.
\end{proof}

\begin{theorem}\label{thm16}
	Let $C^\tau = (V, E, \tau)$ be an edge-periodic cycle on $n$ vertices and $L = \{l_e : e \in E\}$. Then, if $n \geq 2\cdot l \cdot \textsf{LCM}(L)$, $C^\tau$ is robber-win (where $l = 1$ if $\textsf{LCM}(L) \geq 2\cdot \max  L  $, and $l=2$ otherwise).
\end{theorem}

\begin{proof}
		We let $c_t$ and $r_t$ denote the vertex at which \textsf{C} and \textsf{R} are positioned at the start of time step $t$, respectively. Consider now some edge $e \in E(C^\tau)$ and classify its vertices as a `left' and `right' vertex arbitrarily; let  the left vertex of each edge be the right vertex of the following edge in the cycle. We proceed by specifying a strategy for \textsf{R}. Initially, let \textsf{C} choose $c_0$; \textsf{R} should choose $r_0$ to be the vertex antipodal to $c_0$ in $C^\tau$. (If $n$ is odd then \textsf{R} should select $r_0$ to be either of the two vertices that are furthest away from $c_0$; we will refer to both these vertices as antipodal to $c_0$, and treat vertices in all steps $t \geq 0$ in the same way.) We now distinguish between two modes of play, \textit{Hide} and \textit{Escape}, and specify \textsf{R}'s strategy in each of them.
		
		\textbf{Hide mode}: A \textit{Hide period} begins in step~$0$ and in any step $t \geq 2$ such that $c_t$ and $r_t$ are antipodal, but $c_{t-1}$ and $r_{t-1}$ were not. As such, any game in which \textsf{R} follows our strategy begins in a Hide period. The Hide period beginning at step $t$ consists of the steps $t' \in [t, t+x]$ such that $c_{t'}$ and $r_{t'}$ are antipodal, but $c_{t+x+1}$ and $r_{t+x+1}$ are not. Any Hide period is followed directly by an escape period, which will start in step $t+x+1$.
		
		\textsf{R}'s \textbf{Hide strategy}: If the game is in a Hide period during step $t$, \textsf{R} should observe \textsf{C}'s choice of $c_{t+1}$, and always try to move to a vertex antipodal to it. We claim that \textsf{R} cannot be caught in any step belonging to a Hide period. To see this, observe that regardless of whether $\textsf{LCM} = \max  L  $ or $\textsf{LCM}\geq 2\cdot \max  L  $, we have that $n \geq 4\cdot \max  L   \geq 4$. As a result, antipodal vertices in $C^\tau$ are at least distance 2 apart from one another, and the claim follows.
		
		\textbf{Escape mode}: An \textit{Escape period} always begins in a step $t$ such that step $t-1$ was the last step of some Hide period. As such, an Escape period consists of steps $t' \in [t, t+x]$, such that each $c_{t'}$ and $r_{t'}$ are not antipodal, but $c_{t+x+1}$ and $r_{t+x+1}$ are. The last step of the Escape period is then $t+x$, and the first step of the next Hide period is $t+x+1$.
		
		\textsf{R}'s \textbf{Escape strategy}: Assume that some Escape period starts in step $t$. Then, at the start of step $t-1$, $c_{t-1}$ and $r_{t-1}$ were antipodal to one another, and during step $t-1$, we had a situation in which \textsf{C} was able to move towards \textsf{R} in some direction, but the edge incident to $r_{t-1}$ leading in the same direction was not present. Now, recall that if $l = 2$, so that $\textsf{LCM} = \max{L}$, then $n \geq 4 \cdot \textsf{LCM}$; and if $l = 1$ so that $\textsf{LCM} \geq 2 \cdot \max  L  $, then $n \geq 2 \cdot \textsf{LCM}$. Therefore, since $c_{t-1}$ and $r_{t-1}$ are antipodal in $C^\tau$ , if $l = 2$ holds we have that the distance between them is at least $2 \cdot \textsf{LCM}$ and if $l = 1$ holds, the distance between them is at least \textsf{LCM}. Observe now that we are able to view any edge-periodic cycle of finite length as an infinite path whose edge patterns repeat infinitely often. Combining these two facts, it then follows from Lemma \ref{lem15} that when the Escape period starts in step $t$, there exists a strategy for \textsf{R} (which started in the previous step from vertex $r_{t-1}$) which will enable it to stay alive until the Hide period ends. 
		
Finally, since every step $t$ belongs to either a Hide period or an Escape period, we have shown that \textsf{C} can never catch \textsf{R}, and the proof is complete.
\end{proof}

We now give lower bounds on the length required of a strictly edge-periodic cycle to ensure that it is robber-win.

\begin{theorem}\label{thm17}
	There exists an edge-periodic cycle of length $3\cdot \textsf{LCM}$ with edge pattern lengths in the set $L$ that is both cop-win and satisfies $\textsf{LCM} = \max  L  $.
\end{theorem}
\begin{proof}
	Let $M > 1$ be an integer and consider an edge-periodic cycle $C$ of length $3M$ with edge pattern lengths in $L = \{1, M\}$. Let $2$ consecutive edges have patterns $0...01$ of length $M$, and all $3M - 2$ remaining edges have pattern $1$. We refer to the subpath of $C$ consisting of the two edges with period $M$ as the $M$-path, and the subpath with edges labelled with $1$ as the $1$-path. 
	
	We now specify a strategy for \textsf{C} and show that it is in fact a winning strategy: Let \textsf{C} position itself initially at either of the two vertices belonging to the $1$-path that are distance $M-1$ from one extreme point of the $M$-path, and $2M-1$ from the opposite extreme point (where distance is taken to mean the length of the path to that extreme point that avoids the edges of the $M$-path). Call that chosen vertex $c_0$, and notice that it splits the $1$-path into two subpaths that intersect only in $c_0$ -- one of length $M-1$ which we will call $P^-$, the other of length $2M-1$ that we call $P^+$. If \textsf{R} chooses its initial position to be some vertex lying on $P^-$, then \textsf{C} can move along all edges of $P^-$ in the first $M-1$ steps. Since the only way for \textsf{R} to leave $P^-$ without running into \textsf{C} is via the $M$-path, \textsf{C} will catch \textsf{R} in these $M-1$ steps, since no edge of the $M$-path is present until step $M-1$. If \textsf{R} chooses its initial position as some vertex lying on $P^+$, then in the first $2M-1$ steps \textsf{C} can traverse all edges of $P^+$. Again, the only way  for \textsf{R} to leave $P^+$ without encountering \textsf{C} is via the $M$-path. This time \textsf{R} will be able to traverse one edge of the $M$-path, but will be stuck at the middle vertex of the $M$-path until time step $2M-1$. Since, in this step (which is the $2M$-th step), \textsf{C} will be positioned at the vertex that lies on both the $M$-path and $P^+$, \textsf{C} will move first and catch \textsf{R}. It remains to be shown that \textsf{R} will be caught if it chooses the middle vertex of the $M$-path as its start vertex: here, \textsf{R} will not be able to move until step $M-1$, so \textsf{C} should traverse all edges of $P^-$ in the first $M-1$ steps. Then, in step $M-1$, \textsf{C} will be at one endpoint of the $M$-path, with \textsf{R} at the middle vertex of the $M$-path -- in this step, \textsf{C} will move first and catch \textsf{R}. Since we have shown that \textsf{C} wins in all cases, and since $M = \textsf{LCM}(\{1, M\})$, the theorem follows.
\end{proof}

A small amount of modification to the construction in the proof of Theorem \ref{thm17} yields an  additional lower bound for the $\textsf{LCM} \geq 2\cdot \max  L  $ case:

\begin{theorem}\label{thm18}
	There exists an edge-periodic cycle of length $1.5 \cdot \textsf{LCM}$ with edge periods in the set $L$ that is both cop-win and satisfies $\textsf{LCM} \geq 2\cdot \max  L  $.
\end{theorem} 
\begin{proof}
Perform the construction from the proof of Theorem \ref{thm17}, taking $M > 1$ to be odd. Select one of the vertices that has distance $M-1$ and $2M-1$ from opposite ends of the $M$-path, calling that vertex $x$. Consider the strategy from the proof of Theorem \ref{thm17}, and notice that there are two edges that \textsf{C} may cross in the second step. Select either one of these edges and replace its current pattern of $1$ with a pattern of the form $01$ (with length $2$). \textsf{C} should now select its initial vertex as $x$ and follow the strategy in the proof of Theorem \ref{thm17} -- this works since the edge with pattern $01$ has been selected so that it is present whenever \textsf{C}'s strategy crosses that edge. Now, notice that $M$ is odd, and so we have that $\textsf{LCM} = 2M$. Since the constructed cycle has length $3M$, the theorem follows.
\end{proof}

\section{Conclusion}
We considered the problem of deciding whether a single cop can catch a single robber on a given graph, but extended the class of viable game arenas to include the edge-periodic graphs. For a given $n$-vertex edge-periodic graph whose edges are labelled with patterns of lengths in the set $L$, we showed, amongst other things, that there exists an algorithm with running time $O(\textsf{LCM}(L) \cdot n^3)$ that decides whether the cop or robber wins, and computes a strategy for the winning player.
One natural open question that we find particularly interesting is the following: what is the complexity of deciding whether cop or robber wins when the least common multiple of the label lengths, $\textsf{LCM}(L)$, is not bounded by a polynomial? We note that $\textsf{LCM}(\{1,...,n\}) = \textsf{e}^{\phi(n)}$, where $\phi(n) \in \Theta(n)$ is Chebyshev's function \cite{Rankin61}. This shows that, given a graph with at least $n$ edges and at least one edge labelled by each of the integers $i \in \{1,2,\ldots,n\}$, the transformation function described in Section 3.1 would produce a directed graph of order $\textsf{e}^{\Omega(n)}$, i.e., exponential in $n$, and as such the algorithm of Theorem \ref{thm2} would have exponential running time. It would be interesting to establish whether there exists a better algorithm, or whether the problem is $\mathit{NP}$-hard for this case. More generally, one could also examine the cops and robbers problem within the context of other temporal graph models.

In the second part of the paper, we obtained a $2\cdot l \cdot \textsf{LCM}(L)$ upper bound on the length required of any edge-periodic cycle to ensure that it is robber-win, where $l = 1$ if $\textsf{LCM}(L) \geq 2\max  L  $, and $l=2$ otherwise. We also provided lower bound constructions of cop-win cycles that have lengths $1.5\cdot \textsf{LCM}(L)$ and $3\cdot \textsf{LCM}(L)$ in the $l=1$ and $l=2$ cases, respectively. As an example of a potentially interesting further direction, one could attempt to tighten the gap between upper and lower bounds on the minimum length of edge-periodic cycles to be guaranteed to be robber-win in both of these cases.

\section*{Acknowledgements}
The authors would like to thank an anonymous reviewer for a suggestion
leading to the running-time for the variant with $k$ cops mentioned at
the end of Section~\ref{sec:main}.

\bibliographystyle{plain}
\bibliography{cops_robbers_periodic}
\end{document}